\documentclass[twocolumn,superscriptaddress,nofootinbib]{revtex4-1}
\usepackage{geometry}
\usepackage[latin9]{inputenc}
\usepackage[ngerman, english] {babel}  
\usepackage[T1]{fontenc}
\usepackage{fancyhdr}
\usepackage{etoolbox}

\usepackage{titlesec}
\titleformat*{\paragraph}{\bf}

\usepackage{tikz} 
\usetikzlibrary{hobby}
\usepackage{tikz-cd}
\usepackage{graphicx}
\usepackage{caption}
\usepackage{enumerate}
\usepackage{enumitem}
\usepackage{comment}
\usepackage{verbatim}
\usepackage{chngcntr}
\usepackage{apptools}
\usepackage{amsmath}
\usepackage{amssymb}
\usepackage{amsfonts}
\usepackage{amsthm}
\usepackage{eufrak}
\usepackage{mathtools}
\usepackage{cases}
\usepackage{mathrsfs}
\usepackage{bbold}
\usepackage{latexsym}
\usepackage{wasysym}
\usepackage{stmaryrd}
\usepackage{xfrac}
\usepackage{leftidx}
\usepackage{hyperref} 

\newcommand*\diff{\mathop{}\!\mathrm{d}}
\newcommand*\Diff{\mathop{}\!\mathrm{D}}

\theoremstyle{definition}
\newtheorem{definition}{Definition}
\newtheorem*{example}{Example}

\theoremstyle{plain}
\newtheorem{theorem}{Theorem}
\newtheorem{corollary}{Corollary}
\newtheorem{lemma}{Lemma}

\theoremstyle{remark}
\newtheorem*{remark}{Remark}

\begin{document}

\title{Null Lagrangians of non-local field theories}

\author{Kevin Thieme}
\altaffiliation{Current affiliation: Physik-Institut, Universität Zürich, Winterthurerstrasse 190, 8057 Zürich, Schweiz, kevin.thieme@physik.uzh.ch}

\affiliation{Max-Planck-Institut für Gravitationsphysik (Albert-Einstein-Institut),
~\\
Am Mühlenberg~1, 14476~Potsdam-Golm, Deutschland}

\affiliation{Eidgenössische Technische Hochschule Zürich, Departement Physik,
~\\
Otto-Stern-Weg~1, 8093~Zürich, Schweiz}




\begin{abstract}
This manuscript provides a characterisation of the equivalence class of classical smooth Lagrangian densities that involve terms depending on two distinct points of the underlying Euclidean base space of the theory. Theories of this type are referred to as non-local field theories, which are of particular interest in the group field theory approach to quantum gravity. The notion of equivalence of Lagrangian densities is set by physical indistinguishability by means of their equations of motion whose derivation is shown briefly. We expect our results to give a more comprehensive view on the non-local Noether theorem regarding divergence symmetries. 
\end{abstract}

\maketitle

\section{Introduction}

Locality is at the very heart of our conception of fundamental particle interactions, in particular, from the mathematical point of view of contact interaction terms involving contractions of quantum fields evaluated at the same point in spacetime. However, diverging from this fundamental description by integrating out microscopic degrees of freedom, one may encounter non-localities. One of the most evident examples here is the Coulomb interaction between electrically charged particles. Were we not interested in the fundamental mediation by photons nor aware of even their existence, we would describe the phenomenon by an effective interaction ``at a distance'' whose essence is the electrical potential with its well-known inverse distance proportionality. In this regard, a Lagrangian density is referred to as \textit{non-local} (sometimes also \textit{multi-local}) if it is a functional on possibly $N$ copies of a single jet space, assuming that each of the fundamental fields depends on only one point of the base space~\cite{Kegeles2016a}. We will see in the following that non-locality is closely connected with the fact that the corresponding Euler-Lagrange equations are in integro-differential form.

This type of non-locality is in particular encountered in the group field theory approach to quantum gravity~\cite{Oriti2009,Oriti2013, Oriti2014, Oriti2015} in which case it emerges in the combinatorially non-local structure of the quanta~\cite{Oriti2012}. 

A Noether theorem was proposed for theories with Lagrangian densities of the described kind~\cite{Kegeles2016a, Kegeles2016b}. However, it did not cover the analogue of the divergence symmetry present in the well-known local theorem, i.e.~the symmetry on the level of the Lagrangian densities is still unknown. This symmetry would not only complement the non-local Noether theorem but also give rise to the notion of ``equivalent'' Lagrangian densities in the non-local case. Equivalent Lagrangian densities differ by terms that do not affect the equations of motion. These terms are called ``null Lagrangians''. 
For instance, let us suppose we describe a physical system in a connected volume, the experimental domain, of (open) size $\Omega \subset \mathbb{R}$ by a local classical scalar field theory
\begin{equation*}
S_{\mathrm{loc}}=\int \limits_{\Omega} L_{\mathrm{loc}}(x,\psi(x)) \, \diff x.
\end{equation*} 
However, the same system could as well be described by  
\begin{align}
S_{\mathrm{loc + nloc}}=&\int \limits_{\Omega} L_{\mathrm{loc}}(x,\psi(x)) \, \diff x + \notag \\
&\int \limits_{\Omega \times \Omega} L_{\mathrm{nloc}}(x,\phi(x),y,\phi(y)) \, \diff x \, \diff y \notag
\end{align}
if the non-local Lagrangian density $L_{\mathrm{nloc}}$ is null on $\Omega \times \Omega$ for any field $\phi$, i.e.~its equation of motion vanishes identically. 
Based on physical criteria, we thus cannot prefer one of the descriptions over the other. This implies that we can modify any fundamental theory in such a way that it gets formally non-local.

While null Lagrangians had been studied for the local case and identified to be total divergences~\cite{Edelen1962, Ball1981, Olver1983, Olver1985, Olver1986, Olver1988}, we will characterise them for the non-local classical case. For the sake of simplicity, we will restrict ourselves to Lagrangian densities on two copies of a single jet space and will, in the following, simply refer to the ``Lagrangian density'' as ``Lagrangian''. 

Our paper is structured as follows. In Section~\ref{Sec:Non-local_Euler-Lagrange_equations} we derive the non-local equations of motion, that are the Euler-Lagrange equations, using common variational principles. Based on this, we can give a precise definition of a non-local null Lagrangian. Sections~\ref{Sec:Null_Lagrangians_for_specific} and~\ref{Sec:Null_Lagrangians_for_all} are the core of this paper. Here, we prove necessary and sufficient conditions for Lagrangians to be null on a certain domain of the base space of the theory and on all its subdomains, respectively. We summarise and interpret our findings in Section~\ref{Sec:Conclusions}. 

\section{Non-local Euler-Lagrange equations}
\label{Sec:Non-local_Euler-Lagrange_equations}
 
Let $X \simeq \mathbb{R}^{p}$ and $U \simeq \mathbb{R}^{q}$. Consider a Lagrangian as smooth functional on the product of the $n$th and $m$th order jet space of the basic space $X \times U$,
\begin{equation}
L_{12}: \, J^{n} \times J^{m} \rightarrow \mathbb{R},
\end{equation}
where the subscript~$12$ denotes the order of the jet spaces in the Cartesian product\footnote{Introduce the permutation operator $\mathcal{P}: J^{n} \times J^{m} \rightarrow J^{m} \times J^{n}$ such that $L_{21}=\mathcal{P} \circ L_{12}$ and $L_{12}(x,u^{(n)},y,v^{(m)})=L_{21}(y,v^{(n)},x,u^{(m)})$.}.
For an open, connected region $\Omega \subset X$ with smooth boundary $\partial \Omega$, consider a smooth real-valued function $\phi: \, \Omega \rightarrow U$ and the action functional 
\begin{equation}
\label{Eq:Variational_problem}
S[u]=\int \limits_{\Omega \times \Omega} L_{12}(x,u^{(n)},y,v^{(m)}) \, \diff x \, \diff y.
\end{equation}
Its smooth extrema $u=\phi(x)$, $v=\phi(y)$ (with $u^{(n)}$ and $v^{(m)}$ being the $n$th and $m$th prolongation of $\phi$ with respect to $x$ and $y$) are characterised by the Euler-Lagrange equations that follow from a variation with a smooth function $\eta$ with compact support in $\Omega$, i.e.~$\left. \eta \right|_{\partial \Omega}=0$ (cf.~\cite[p.~249f.]{Olver1986} for the local case),
\begin{equation}
0\overset{!}{=}\left. \frac{\mathrm{d}}{\mathrm{d} \varepsilon} S[\phi+\varepsilon\eta] \right|_{\varepsilon=0} \label{Eq:Derivation_EL_equations}
\end{equation}
for $\varepsilon \in \mathbb{R}$. The fundamental lemma of the calculus of variations yields the Euler-Lagrange equations for our non-local theory,
\begin{align}
\label{Eq:Euler_Lagrange_equation}
&\int \limits_{\Omega} \left\{ E^{1}(L_{12})\left(y, v^{(2n)}, x, u^{(m)}\right)+ \right. \notag\\
&\hphantom{\int \limits_{\Omega} ..} \left. E^{2}(L_{12})\left(x, u^{(n)}, y, v^{(2m)}\right) \right\} \, \diff x =0.
\end{align}
Here, the Euler operator corresponding to the first jet space is
\begin{equation}
\label{Def:Euler_operator}
E^{1} \coloneqq \sum_{\alpha=1}^{q}\sum_{J}(-\Diff^{1})_{J} \frac{\partial}{\partial u_{J}^{\alpha}},
\end{equation}
where we defined
\begin{equation}
(-\Diff^{1})_{J}\coloneqq(-1)^{k}\Diff_{J}^{1}\coloneqq(-\Diff^{1}_{j_{1}})\cdots(-\Diff^{1}_{j_{k}})
\end{equation}
for a multi-index $J=(j_{1}, \ldots, j_{k})$ of order $k$ with entries ranging from $1$ to $p$ and the total derivative with respect to coordinate $i$ is~\cite[p.~112]{Olver1986}
\begin{equation}
\label{Eq:DefTotDer}
\Diff_{i}^{1}\coloneqq \frac{\partial}{\partial x^{i}}+\sum_{\alpha=1}^{q}\sum_{J}u_{J,i}^{\alpha} \frac{\partial}{\partial u_{J}^{\alpha}}.
\end{equation} 
Note that our notation enables us to treat the action of $E^{1}$ on $L_{12}$ as an action of $E^{2}$ on $L_{21}$ and vice versa, 
\begin{align}
\label{Eq:E1L12_E2L21}
E^{1}(L_{12})\left(y,v^{(2n)},x,u^{(m)}\right)&= \notag \\
E^{2}(L_{21})&\left(x,u^{(m)},y,v^{(2n)}\right),
\end{align}
or
\begin{align}
\label{Eq:E2L12_E1L21}
E^{2}(L_{12})\left(x,u^{(n)},y,v^{(2m)}\right)&= \notag \\
E^{1}(L_{21})&\left(y,v^{(2m)},x,u^{(n)}\right),
\end{align}
respectively. Hence, using the linearity of the Euler operator, the integrand of Eq.~\eqref{Eq:Euler_Lagrange_equation} becomes the left side of
\begin{equation}
\label{Eq:Antisymcond1}
E^{2}(L_{12}+L_{21})\left(x,u^{(\max\{n,m\})},y,v^{(\max\{2n,2m\})}\right)=0.
\end{equation} 
Denoting shortly an integration of an integrable function $f_{12}$ over a region $\Omega \subset X$ corresponding to the first jet space by $\mathcal{I}_{\Omega}^{1}(f_{12})(y,v^{(m)}) \coloneqq \int \limits_{\Omega} f_{12}(x,u^{(n)},y,v^{(m)}) \, \diff x$, the Euler-Lagrange equations~\eqref{Eq:Euler_Lagrange_equation} appear
\begin{equation}
\label{Eq:Euler_Lagrange_equation_new}
\mathcal{I}_{\Omega}^{1}(E^{2}(L_{12}+L_{21})) \left(y,v^{(\max\{2n,2m\})} \right)=0.
\end{equation}
Written in this form, it is apparent that all antisymmetric terms of the Lagrangian (under the action of $\mathcal{P}$) vanish identically in the integrand and that all observable effects are defined by its symmetric parts.

The aim of this paper is to find those functionals $L_{12}$ that satisfy the Euler-Lagrange equations independently of the field configuration $u$,~$v$, i.e.~we are looking for Lagrangians for which every function $u=\phi(x)$, $v=\phi(y)$ extremises the functional~\eqref{Eq:Variational_problem}. These are called \textit{null Lagrangians}:
\begin{definition}
A Lagrangian $L_{12}=L_{12}(x,u^{(n)},y,v^{(m)})$ is \textit{null on $\Omega$} if it satisfies Eq.~\eqref{Eq:Euler_Lagrange_equation} for all $u$,~$v$ and their various derivatives at all points $x$,~$y$ in $\Omega$. Any two Lagrangians $L'_{12}$ and $L''_{12}$ are called \textit{equivalent} ($L'_{12} \sim L''_{12}$) if they differ by a null Lagrangian $L_{12}=L'_{12}-L''_{12}$.
\end{definition} 

\section{Null Lagrangians on particular $\boldsymbol{\Omega}$}
\label{Sec:Null_Lagrangians_for_specific}

First, we investigate Eq.~\eqref{Eq:Euler_Lagrange_equation_new} in its complete form to characterise null Lagrangians on a certain domain $\Omega$. In general, the fields will have to satisfy specific boundary conditions and a null Lagrangian will be tailored to that domain and not be null on its subdomains.
\begin{theorem}
\label{Th:Nulllagrangian4}
A non-local Lagrangian $L_{12}=L_{12}(x,u^{(n)},y,v^{(m)})$ is null on $\Omega$ if and only if there is a $p$-tuple $P$ of smooth functions such that
\begin{equation}
\label{Eq:Nulllagrangian_nloc_4}
\mathcal{I}_{\Omega}^{1}(L_{12}+L_{21})=\mathrm{Div} \, P (y,v^{(\max\{n,m\})}).
\end{equation}
\end{theorem}
\begin{remark}
The total divergence is given by
\begin{equation}
\mathrm{Div} \, P=\Diff_{1}P_{1}+\cdots+\Diff_{p}P_{p}.
\end{equation}
\end{remark}
\begin{lemma}
\label{Le:Exchange_IE}
We can exchange 
\begin{equation}
\mathcal{I}_{\Omega}^{k}(E^{l}(L_{12}))=E(\mathcal{I}_{\Omega}^{k}(L_{12})), \quad k,l \in \{ 1,2 \}, \, k \neq l. 
\end{equation}
\end{lemma}
\begin{proof}
Consider the definition of the Euler operator~\eqref{Def:Euler_operator} and of the total derivative~\eqref{Eq:DefTotDer}. For continuously differentiable functions, in particular for the Lagrangian, integrals and partial derivatives with respect to coordinates that are independent of the coordinate of integration commute. Coordinates like field derivatives in Eq.~\eqref{Eq:DefTotDer} that are independent of the coordinate of integration behave like constants. 
\end{proof}
\begin{proof}[Proof of Theorem~\ref{Th:Nulllagrangian4}]
The claim follows immediately from Lemma~\ref{Le:Exchange_IE} using the linearity of the Euler operator and taking account of the equivalence~\cite[p.~252]{Olver1986}:
\begin{equation}
\label{Eq:Olvers_equivalence}
E(L)=0 \quad \Leftrightarrow \quad \exists \, P: \, L=\mathrm{Div} \, P
\end{equation}
for a local Lagrangian defined on $X \times U^{(n)}$ and a $p$-tuple~$P$ of smooth functions depending on $x,u$ and derivatives of $u$.
\end{proof}

The following corollary characterises the subset of null Lagrangians on $\Omega$ for which both summands in Eq.~\eqref{Eq:Euler_Lagrange_equation_new} vanish separately.
\begin{corollary}
\label{Th:Nulllagrangian3}
A non-local Lagrangian $L_{12}=L_{12}(x,u^{(n)},y,v^{(m)})$ satisfies
\begin{equation}
\label{Eq:Euler_Lagrange_equation2_new}
\mathcal{I}_{\Omega}^{2}(E^{1}(L_{12})) =
\mathcal{I}_{\Omega}^{1}(E^{2}(L_{12})) =0
\end{equation}
for all $x,u,y,v$ and derivatives of $u$ and $v$ if and only if there are $p$-tuples $P^{1}$,~$P^{2}$ of smooth functions such that
\begin{equation}
\label{Eq:Nulllagrangian_nloc_3_1}
\mathcal{I}_{\Omega}^{2}(L_{12})=\mathrm{Div} \, P^{1} (y,u_{2}^{(n)}), 
\end{equation}
and 
\begin{equation}
\label{Eq:Nulllagrangian_nloc_3_2}
\mathcal{I}_{\Omega}^{1}(L_{12})=\mathrm{Div} \, P^{2} (y,u_{2}^{(m)}).
\end{equation} 
\end{corollary}
\begin{proof}
This is a special case of Theorem~\ref{Th:Nulllagrangian4}.
\end{proof}
The null Lagrangians characterised by Theorem~\ref{Th:Nulllagrangian4} and Corollary~\ref{Th:Nulllagrangian3} depend on the choice of the domain~$\Omega$. This prevents the deduction of an explicit expression for $L_{12}$. However, we have been able to eliminate the Euler operator to obtain direct conditions for the Lagrangians which characterise the equivalence class of non-local Lagrangians. 

\section{Null Lagrangians on arbitrary $\boldsymbol{\Omega' \subset \Omega}$}
\label{Sec:Null_Lagrangians_for_all}

If we demanded the Lagrangians to be null on any domain and in particular on any subdomain $\Omega' \subset \Omega$, already the integrand of the Euler-Lagrange equations~\eqref{Eq:Euler_Lagrange_equation} would need to vanish. The resulting conditions would be, thus, independent of the region under investigation and not involve boundary conditions. First, we will study the case when both summands vanish separately. The corresponding null Lagrangians, which will take a closed form here, are characterised by the following theorem.
\begin{theorem}
\label{Th:Nulllagrangian1}
A non-local Lagrangian $L_{12}=L_{12}(x,u^{(n)},y,v^{(m)})$ satisfies 
\begin{equation}
\label{Eq:Simple_EL2}
E^{1}(L_{12})=E^{2}(L_{12})=0 
\end{equation}
for all $x,u,y,v$ and derivatives of $u$ and $v$ if and only if there exists an antisymmetric $2p \times 2p$-matrix $Q_{12}$ with entries being smooth functions of $x,u,y,v$ and derivatives of $u$ and $v$ such that
\begin{equation*}
L_{12}=\frac{1}{2}\sum_{i,j=1}^{2p} \mathfrak{D_{i}}\mathfrak{D_{j}}(M \cdot Q_{12})_{ij}, 
\end{equation*}
\begin{equation}
\label{Eq:Nulllagrangian_nloc_1}
\mathfrak{D}_{i} \coloneqq 
\begin{cases}
\Diff_{i}^{1}, \,\,\,\,\; \quad i \in \{1,\ldots,p\} \\
\Diff_{i-p}^{2}, \quad i \in \{p+1,\ldots,2p\}
\end{cases},
\end{equation}
where $M=\mathrm{diag}(\mathbb{1}_{p},-\mathbb{1}_{p})$.
\end{theorem}
\begin{proof}
Note that Eq.~\eqref{Eq:Simple_EL2} is equivalent to (cf.~statement~\eqref{Eq:Olvers_equivalence})
\begin{align}
\label{Eq:Div1andDiv2condition}
L_{12}=&\mathrm{Div}_{1} \, P^{1}(x,u^{(n)},y,v^{(m)}) \notag\\
=&\mathrm{Div}_{2} \, P^{2}(x,u^{(n)},y,v^{(m)}),
\end{align}
where
\begin{equation}
\mathrm{Div}_{k} \, P^{k}=\Diff_{1}^{k}P_{1}^{k}+\cdots+\Diff_{p}^{k}P_{p}^{k}, \quad k \in \{ 1,2 \}
\end{equation}
for $p$-tuples $P^{1}$ and $P^{2}$ of smooth functions. We write the total divergence on $\mathbb{R}^{p} \oplus \mathbb{R}^{p}$ as
\begin{equation}
\mathfrak{Div} \, (P^{1},P^{2}) \coloneqq \mathrm{Div}_{1} \, P^{1}+ \mathrm{Div}_{2} \, P^{2}.
\end{equation}
Thus, the two conditions~\eqref{Eq:Div1andDiv2condition} can equivalently be written as 
\begin{equation}
\label{Eq:Nullcondition1}
L_{12}=\frac{1}{2}\mathfrak{Div} \, (P^{1},P^{2})
\end{equation}
for
\begin{equation}
\label{Eq:Nullcondition2}
\mathfrak{Div} \, (P^{1},-P^{2})=0.
\end{equation}
``$\Rightarrow$'': \\
Consider the involutory $2p \times 2p$-matrix $M=\mathrm{diag}(\mathbb{1}_{p},-\mathbb{1}_{p})$. The $2p$-tuple $(P^{1},-P^{2})=M \cdot (P^{1},P^{2})$ is in the kernel of $\mathfrak{Div}$ if and only if there exists a matrix $Q_{12}$ as described in Theorem~\ref{Th:Nulllagrangian1} such that~\cite{Olver1983},~{\cite[p.~269]{Olver1986}}  
\begin{equation}
\label{Eq:Kernel_Divergence}
(P^{1},P^{2})_{i}=\sum_{j=1}^{2p} \mathfrak{D}_{j}(M \cdot Q_{12})_{ij}.
\end{equation}
Hence, we can conclude in the face of Eq.~\eqref{Eq:Nullcondition1}
\begin{equation}
L_{12}=\frac{1}{2}\sum_{i,j=1}^{2p} \mathfrak{D}_{i}\mathfrak{D}_{j}(M \cdot Q_{12})_{ij}.
\end{equation}
``$\Leftarrow$'': \\
By definition, $L_{12}$ in Eq.~\eqref{Eq:Nulllagrangian_nloc_1} is of the form in condition~\eqref{Eq:Nullcondition1} with the tuples 
\begin{equation}
P^{1}=\left( \sum_{j=1}^{2p} \mathfrak{D}_{j} (Q_{12})_{ij} \right)_{i=1,\ldots,p},
\end{equation}
\begin{equation}
P^{2}=\left( \sum_{j=1}^{2p} \mathfrak{D}_{j} (-Q_{12})_{ij} \right)_{i=p+1,\ldots,2p}
\end{equation}
which in turn satisfy condition~\eqref{Eq:Nullcondition2} because of the if-and-only-if statement about null divergences~\eqref{Eq:Kernel_Divergence}.  
\end{proof}
\begin{example}
Consider the kinetic term $L_{12}=\frac{1}{2} \partial_{x}\phi(x) \partial_{y}\phi(y)$ that can be written in the form of Eq.~\eqref{Eq:Nulllagrangian_nloc_1} with
\begin{equation*}
Q_{12}=\frac{1}{2}\phi(x) \phi(y)\left( 
\begin{array} {cc}
0 & 1 \\
-1 & 0
\end{array} \right)
\end{equation*}
and, thus, it is null.
\end{example}
Finally, we will characterise the subset of null Lagrangians that satisfy the full integrand~\eqref{Eq:Antisymcond1} of the Euler-Lagrange equations which clearly includes Lagrangians of the form~\eqref{Eq:Nulllagrangian_nloc_1}. 
\begin{theorem}
\label{Th:Nulllagrangian2}
A non-local Lagrangian $L_{12}=L_{12}(x,u^{(n)},y,v^{(m)})$ is null on any $\Omega' \subset \Omega$ if and only if there exists an antisymmetric $2p \times 2p$-matrix $Q_{12}$ with entries being smooth (and symmetric) functions $(Q_{12})_{ij}=(Q_{21})_{ij}$ of $x,u,y,v$ and derivatives of $u$ and $v$ such that
\begin{equation}
\label{Eq:Nulllagrangian_nloc_2}
(L_{12}+L_{21})(x,u^{(n)},y,v^{(m)})=\sum_{i,j=1}^{2p} \mathfrak{D_{i}}\mathfrak{D_{j}}(M \cdot Q_{12})_{ij},
\end{equation}
where $M=\mathrm{diag}(\mathbb{1}_{p},-\mathbb{1}_{p})$.
\end{theorem}
\begin{proof}
Again, using the local statement~\eqref{Eq:Olvers_equivalence} and repeating the argumentation of the proof of Theorem~\ref{Th:Nulllagrangian1}, Eq.~\eqref{Eq:Antisymcond1} is equivalent to Eq.~\eqref{Eq:Nulllagrangian_nloc_2}. Note that the left-hand side of Eq.~\eqref{Eq:Nulllagrangian_nloc_2} is invariant under permutation with $\mathcal{P}$ which must then also apply to the right-hand side.
\end{proof}

\section{Conclusions}
\label{Sec:Conclusions}

Applying calculus of variations, we have derived the Euler-Lagrange equations for a non-local action corresponding to two points of a region in Euclidean space. The nature of non-locality is reflected in the fact that those are integro-differential equations. A generalisation of this derivation can be found in~\cite{Kegeles2016a}. 

Null Lagrangians are defined to be those functionals that satisfy the Euler-Lagrange equations for all points of the product of the underlying jet spaces. From a physical point of view, two theories on a shared domain are equivalent and indistinguishable if the defining two Lagrangians differ by a null Lagrangian on the very domain. 

By the use of the known local statement that classifies null Lagrangians to be divergences, we have derived integral equations for a fixed domain equivalent to the Euler-Lagrange equations whereby null Lagrangians can either be constructed or identified as null. In addition, we have given explicit expressions for Lagrangians that are null on any subdomain. The introductory discussion has briefly illustrated an interesting physical implication: local problems might be equivalently described by non-local terms. If a system features special conditions for field values on the boundary of the considered region $\Omega$, like periodicity, it is not possible to distinguish a local and non-local description if the non-local terms are null. However, on a different region $\Omega'$ which is, compared to $\Omega$, deformed, resized or equipped with different boundary conditions, the non-local terms will in general contribute to the equations of motion and be non-null. Then the class of null Lagrangians is reduced to the set of Lagrangians that are null on any $\Omega'$ and are of a higher-dimensional divergence form.  

This paper completes the former work~\cite{Kegeles2016a,Kegeles2016b} regarding divergence symmetries. While we have worked out the simplest non-local case, a generalisation to arbitrary non-local theories on possibly $N$ copies of a single jet space is straightforward. An extension to non-Euclidean space would be interesting and would provide an inclusion of divergence symmetries into Noether's theorem in full generality.    

\section*{Acknowledgements}

This work was part of my Master's thesis. I thank Alexander Kegeles for his help, patience and all the fruitful afternoons we spent discussing the topic and its technical issues. I am deeply grateful for his careful and thorough revisions of my notes and his great contributions to the work. Moreover, I thank the entire Quantum Gravity division at AEI~Potsdam, and notably Daniele~Oriti for advising me during the project.

\onecolumngrid

\end{document}